\newif\ifICALP\ICALPtrue

\documentclass[11pt,onecolumn]{article}

\usepackage[left=1.0in,top=1.0in,right=1.05in,bottom=1.0in,nohead,textheight=10in,footskip=0.3in]{geometry}

\usepackage{amsfonts,amsmath,epsfig}

\usepackage{amsthm,amssymb,amsfonts,epsfig,graphicx}
\newtheorem{theorem}{Theorem}
\newtheorem{corollary}[theorem]{Corollary}

\newtheorem{definition}[theorem]{Definition}
\newtheorem{proposition}[theorem]{Proposition}

\newcommand{\myparagraph}[1]{{\vspace*{2pt}\noindent\bf{#1}~~}}

\def\D{{\cal D}}

\long\def\ignore#1{}
\def\myps[#1]#2{\includegraphics[#1]{#2}}
\def\etal{{\em et al.}}

\def\br(#1,#2){{\langle #1,#2 \rangle}}
\def\setZ[#1,#2]{{[ #1 .. #2 ]}}

\def\q={\quad=\quad}
\def\qq={\qquad=\qquad}

\def\calP{{\cal P}}
\def\calR{{\cal R}}

\def\psfile[#1]#2{}
\def\psfilehere[#1]#2{}

\def\assign(#1,#2){\langle#1,#2\rangle}
\def\edge(#1,#2){(#1,#2)}

\def\slack(#1){\texttt{slack}({#1})}
\def\barslack(#1){\overline{\texttt{slack}}({#1})}

\def\unitvec(#1){{{\bf u}_{#1}}}

\def\sign{\texttt{sign~\!\!\!}}
\def\INWARD{\texttt{INWARD}}
\def\OUTWARD{\texttt{OUTWARD}}
\def\NEIB{\texttt{NEIB}}

\newcommand{\bx}{\mbox{\boldmath $x$}}
\newcommand{\by}{\mbox{\boldmath $y$}}
\newcommand{\bz}{\mbox{\boldmath $z$}}

\newcommand{\bX}{\mbox{\boldmath $X$}}
\newcommand{\bY}{\mbox{\boldmath $Y$}}

\newcommand{\bu}{\mbox{\boldmath $u$}}
\newcommand{\bv}{\mbox{\boldmath $v$}}

\begin{document}
%
%
%
%
%
%
\title{Submodularity on a tree: \\ Unifying $L^\natural$-convex and bisubmodular functions}
%
%
\author{Vladimir Kolmogorov \\ ~ \\ \normalsize University College London \\ {\normalsize\tt v.kolmogorov@cs.ucl.ac.uk} }
\date{}
%
%
%

\maketitle              

\begin{abstract}
We introduce a new class of functions that can be minimized in polynomial time in the value oracle model.
These are functions $f$ satisfying $f(\bx)+f(\by)\ge f(\bx \sqcap \by)+f(\bx \sqcup \by)$
where the domain of each variable $x_i$ corresponds to nodes of a rooted binary tree, and 
operations $\sqcap,\sqcup$ are defined with respect to this tree.
Special cases include previously studied $L^\natural$-convex and bisubmodular functions,
which can be obtained with particular choices of trees.
We present a polynomial-time algorithm for minimizing functions in the new class.
It combines Murota's steepest descent algorithm for $L^\natural$-convex functions
with bisubmodular minimization algorithms.

\end{abstract}
\section{Introduction}\label{sec:intro}
Let $f:\D\rightarrow \mathbb R$ be a function of $n$ variables $\bx=(x_1,\ldots,x_n)$ where $x_i\in D_i$;
thus $\D=D_1\times\ldots\times D_n$. We call elements of $D_i$ {\em labels}, and the argument of $f$ a {\em labeling}. 
Denote $V=\{1,\ldots,n\}$ to be the set of nodes. 
We will consider functions $f$ satisfying
\begin{equation}
f(\bx)+f(\by)\ge f(\bx \sqcap \by) + f(\bx \sqcup \by) \qquad \forall \bx,\by\in \D
\label{eq:submodularity}
\end{equation}
where binary operations $\sqcap,\sqcup:\D\times \D\rightarrow \D$ (expressed component-wise via operations $\sqcap,\sqcup:D_i\times D_i\rightarrow D_i$)
are defined below.

There are several known cases in which function $f$ can be minimized in polynomial time in the value oracle model.
The following two cases will be of particular relevance:
\begin{itemize}
\item {\em $L^\natural$-convex functions\footnote{Pronounced as ``L-natural convex''.}}: $D_i=\{0,1,\ldots,K_i\}$ where $K_i\ge 0$ is integer, 
$a\sqcap b=\lfloor \frac{a+b}{2} \rfloor$,
$a\sqcup b=\lceil \frac{a+b}{2} \rceil$. Property~\eqref{eq:submodularity} is then called {\em discrete midpoint convexity}~\cite{Murota:book}.
\item {\em Bisubmodular functions}: $D_i=\{-1,0,+1\}$, $a\sqcup b=\sign(a+b)$, 
$a\sqcap b=|ab|\sign(a+b)$.
\end{itemize}

In this paper we introduce a new class of functions which includes the two classes above as special cases.
We assume that labels in each set $D_i$ are nodes of a tree $T_i$ with a designated root $r_i\in D_i$.
Define a partial order $\preceq$ on $D_i$ as follows: $a\preceq b$ if $a$ is an ancestor of $b$,
i.e.\ $a$ lies on the path from $b$ to $r_i$ ($a,b\in D_i$). 
For two labels $a,b\in D_i$ let $\calP[a\rightarrow b]$ be unique path from $a$ to $b$ in $T_i$,
$\rho(a,b)$ be the number of edges in this path, and $\calP[a\rightarrow b,d]$ for integer $d\ge 0$ be the $d$-th node of this path
so that $\calP[a\rightarrow b,0]=a$ and $\calP[a\rightarrow b,\rho(a,b)]=b$. If $d>\rho(a,b)$ then we set by definition $\calP[a\rightarrow b,d]=b$.

With this notation, we can now define $a\sqcap b$, $a\sqcup b$
as the unique pair of labels satisfying the following two conditions:
(1) $\{a\sqcap b,a\sqcup b\}=\{\calP[a\rightarrow b,\lfloor\frac{d}{2}\rfloor],\calP[a\rightarrow b,\lceil\frac{d}{2}\rceil]\}$
where $d=\rho(a,b)$, and (2) $a\sqcap b\preceq a\sqcup b$ (Figure~\ref{fig:trees}(a)).
We call functions $f$ satisfying condition~\eqref{eq:submodularity} with such choice of $(\D,\sqcap,\sqcup)$ {\em strongly tree-submodular}.
Clearly, if each $T_i$ is a chain with nodes $0,1,\ldots,K$ and $0$ being the root (Figure~\ref{fig:trees}(b)) then strong tree-submodularity is equivalent
to $L^\natural$-convexity. Furthermore, if each $T_i$ is the tree shown in Figure~\ref{fig:trees}(c) then strong tree-submodularity is equivalent to bisubmodularity.

The main result of this paper is the following
\begin{theorem}
If each tree $T_i$ is binary, i.e.\ each node has at most two children, then a strongly tree-submodular function $f$
can be minimized in time polynomial in $n$ and $\max_i |D_i|$.
\label{th:polynomiality}
\end{theorem}

\begin{figure*}[!t]
\begin{center}
\begin{tabular}{c@{\hspace{50pt}}c}
\includegraphics[scale=0.32]{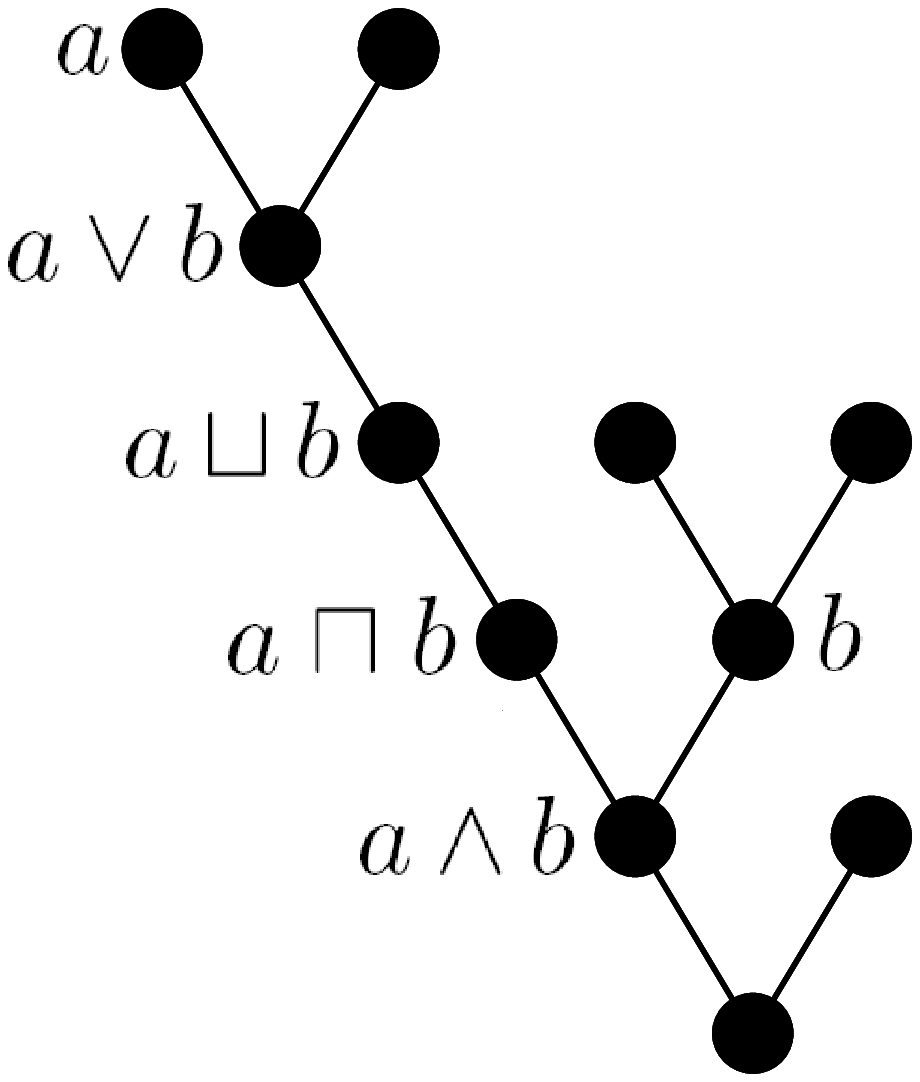} & \includegraphics[scale=0.32]{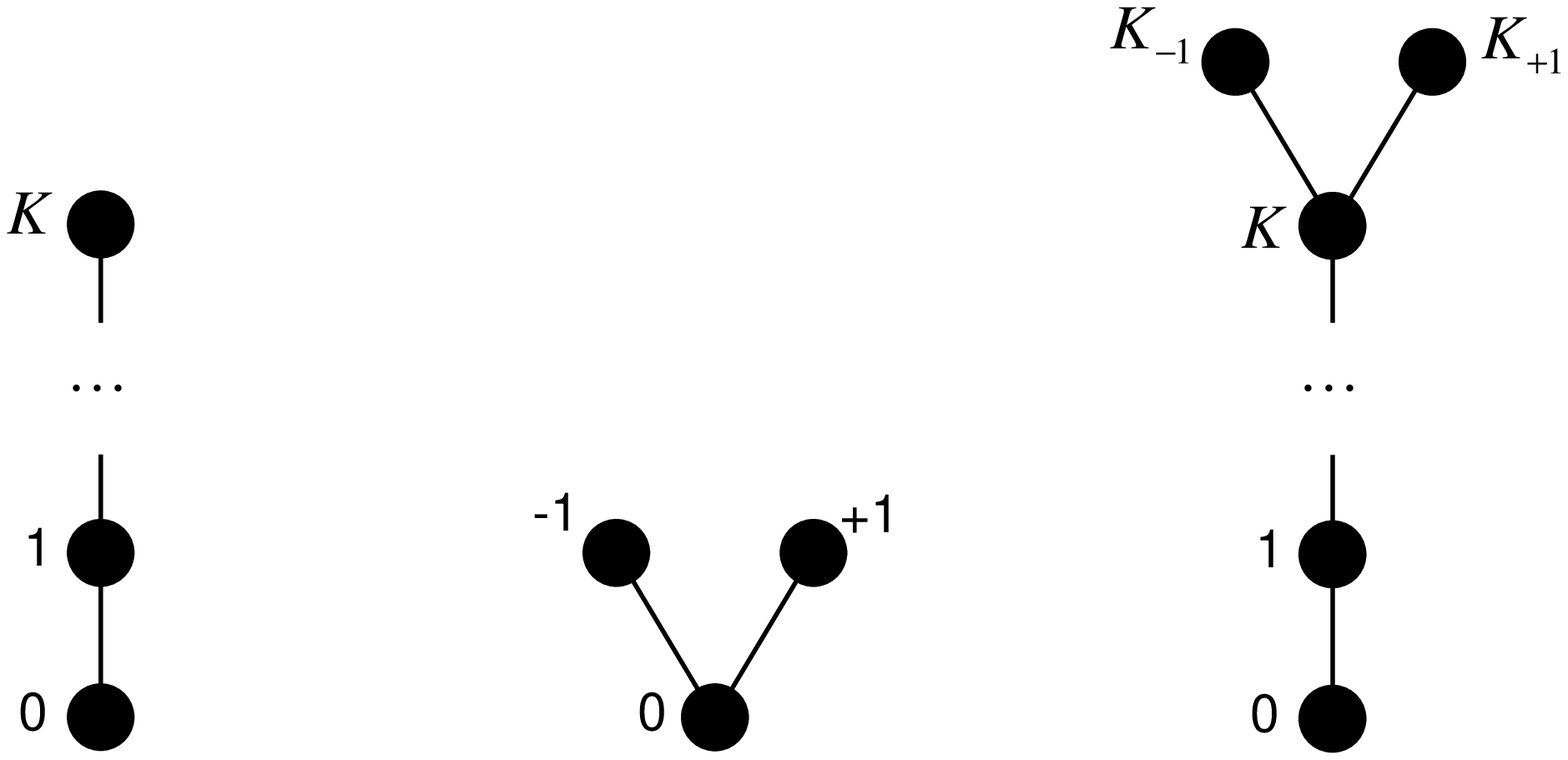} \vspace{-10pt} \\ 
\hspace{30pt} (a) & (b) \hspace{50pt} (c) \hspace{50pt} (d) \hspace{28pt} \vspace{-10pt}
\end{tabular}
\end{center}
\caption{{\bf Examples of trees. Roots are always at the bottom.} (a) Illustration of the definition of $a\sqcap b$, $a\sqcup b$, $a\wedge b$ and $a\vee b$. 
(b) A tree for $L^\natural$-convex functions. (c) A tree for bisubmodular functions. (d) A tree for which a weakly tree-submodular function can be minimized efficiently (see section~\ref{sec:weak}).
}
\label{fig:trees}
\end{figure*}

\myparagraph{Weak tree-submodularity} We will also study alternative operations on trees, which we denote as $\wedge$ and $\vee$.
For labels $a,b\in D_i$ we define $a\wedge b$ as their highest common ancestor, i.e.\ the unique node on the path $\calP[a\rightarrow b]$
which is an ancestor of both $a$ and $b$. The label $a\vee b$ is defined as the unique label on the path $\calP[a\rightarrow b]$
such that the distance between $a$ and $a\vee b$ is the same as the distance between $a\wedge b$ and $b$ (Figure~\ref{fig:trees}(a)).

We say that function $f$ is {\em weakly tree-submodular} if it satisfies
\begin{equation}
f(\bx)+f(\by)\ge f(\bx \wedge \by) + f(\bx \vee \by) \qquad \forall \bx,\by\in \D
\label{eq:submodularity:weak}
\end{equation}
We will show that strong tree-submodularity~\eqref{eq:submodularity} implies weak tree-submodularity \eqref{eq:submodularity:weak},
which justifies the terminology. If all trees are chains shown in Figure~\ref{fig:trees}(b) ($D_i=\{0,1,\ldots,K\}$ with $0$ being the root)
then $\wedge$ and $\vee$ correspond to the standard operations ``meet'' and ''join'' ($\min$ and $\max$) on an integer lattice.
It is well-known that in this case weakly tree-submodular functions can be minimized in time polynomial in $n$ and $K$~\cite{Topkis:78,Murota:book}.
In section~\ref{sec:weak} we give a slight generalization of this result; namely, we allow trees shown in Figure~\ref{fig:trees}(d).

\subsection{Related work} 
Studying operations $\langle\sqcap,\sqcup\rangle$ that give rise to tractable optimization problems received a considerable attention in the literature. Some known examples of
such operations are reviewed below. For simplicity, we assume that domains $D_i$ (and operations $\langle\sqcap,\sqcup\rangle$) are the same for all nodes:
$D_i=D$ for some finite set $D$.

\myparagraph{Submodular functions on lattices} 
The first example that we mention is the case when $D$ is a distributive lattice and $\sqcap,\sqcup$ 
are the meet and joint operations on this lattice. Functions that satisfy~\eqref{eq:submodularity} for this choice of $D$ and $\sqcap,\sqcup$ 
are called {\em submodular functions} on $D$; it is well-known that they can be minimized in strongly polynomial time~\cite{Grotschel:88,Schrijver:00,Iwata:01}.

Recently, researchers considered submodular functions on non-distributive lattices.
It is known that a lattice is non-distributive if it contains as a sublattice either the pentagon ${\cal N}_5$
or the diamond ${\cal M}_3$.
Krokhin and Larose~\cite{KrokhinLarose:08} proved tractability for the pentagon case, using nested applications of a submodular minimization algorithm.
The case of the diamond was considered by Kuivinen~\cite{Kuivinen:TR}, who proved pseudo-polynomiality of the problem.
The case of general non-distributive lattices is still open.

\myparagraph{$L^\natural$-convex functions}
The concept of $L^\natural$-convexity was introduced by Fujishige and Murota~\cite{FujishigeMurota:00} as a variant of $L$-convexity by Murota~\cite{Murota:98}. $L^\natural$-convexity is equivalent to the combination of submodularity and integral 
convexity~\cite{FavatiTardella:90} (see \cite{Murota:book} for details).

The fastest known algorithm for minimizing $L^\natural$-convex functions is the {\em steepest descent} algorithm
of Murota~\cite{Murota:IEICE00,Murota:book,Murota:SIAM03}. 
Murota proved in~\cite{Murota:SIAM03} that algorithm's complexity is $O(n\min\{K,n\log K\}\cdot {\tt SFM(n)})$
where $K=\max_i |D_i|$ and ${\tt SFM}(n)$ is the complexity of a submodular minimization algorithm for a function with $n$ variables.
The analysis of Kolmogorov and Shioura~\cite{KS:09} improved the bound to $O(\min\{K,n\log K\}\cdot {\tt SFM(n)})$.
In section~\ref{sec:alg} we review Murota's algorithm (or rather its version without scaling that has complexity $O(K\cdot {\tt SFM}(n))$.)

Note, the class of $L^\natural$-convex functions is a subclass of submodular functions on a totally ordered set $D=\{0,1,\ldots,K\}$.

\myparagraph{Bisubmodular functions}
Bisubmodular functions were introduced by Chandrasekaran and Kabadi
as rank functions of {\em (poly-)pseudomatroids}~\cite{Chandrasekaran:88,Kabadi:90}.
Independently, Bouchet~\cite{Bouchet:87} introduced the concept of $\Delta$-matroids which is equivalent
to pseudomatroids.
Bisubmodular functions and their generalizations have also been considered by Qi~\cite{Qi:88},
Nakamura~\cite{Nakamura:88}, Bouchet and Cunningham~\cite{Bouchet:95} and Fujishige~\cite{Fujishige:91}.

It has been shown that some submodular minimization algorithms can be generalized to bisubmodular functions.
Qi~\cite{Qi:88} showed the applicability of the ellipsoid method.
Fujishige and Iwata~\cite{Fujishige:06} developed a weakly polynomial combinatorial algorithm for minimizing bisubmodular functions with
complexity $O(n^5 {\tt EO} \log M)$ where ${\tt EO}$ is the number of calls to the evaluation oracle and $M$ is an upper bound on function values.
McCormick and Fujishige~\cite{McCormick:10} presented a strongly combinatorial version with
complexity $O(n^7 {\tt EO} \log n)$, as well as a $O(n^9 {\tt EO} \log^2 n)$ fully combinatorial variant that does not use divisions.
The algorithms in \cite{McCormick:10} can also be applied for minimizing a bisubmodular
function over a {\em signed ring family}, i.e.\ a subset $\calR\subseteq \D$ closed under $\sqcap$ and $\sqcup$.

\myparagraph{Valued constraint satisfaction and multimorphisms}
Our paper also fits into the framework of {\em Valued Constraint Satisfaction Problems} (VCSPs)~\cite{Cohen:AI06}.
In this framework we are given a {\em language} $\Gamma$, i.e.\ a set of cost functions $f:D^m\rightarrow\mathbb R_+\cup\{+\infty\}$
where $D$ is a fixed discrete domain and $f$ is a function of arity $m$ (different functions $f\in\Gamma$ may have different arities).
A {\em $\Gamma$-instance} is any function $f:D^n\rightarrow \mathbb R_+\cup\{+\infty\}$ that can be expressed as a finite sum of functions from $\Gamma$:
$$
f(x_1,\ldots,x_n)=\sum_{t\in T} f_t(x_{i(t,1)},\ldots,x_{i(t,m_t)})
$$
where $T$ is a finite set of terms, $f_t\in \Gamma$ is a function of arity $m_t$, and $i(t,k)$ are indexes in $\{1,\ldots,n\}$.
A finite language $\Gamma$ is called {\em tractable} if any $\Gamma$-instance can be minimized in polynomial time,
and {\em NP-hard} if this minimization problem is NP-hard. These definitions are extended to infinite languages $\Gamma$ 
as follows: $\Gamma$ is called tractable if any finite subset $\Gamma'\subset\Gamma$ is tractable,
and NP-hard if there exists a finite subset $\Gamma'\subset\Gamma$ which is NP-hard.

Classifying the complexity of different languages has been an active research area.
A major open question in this line of research is the
\emph{Dichotomy Conjecture} of Feder and Vardi (formulated for the {\em crisp} case), which states that every
constraint language is either tractable or NP-hard~\cite{Feder98:monotone}.
So far such dichotomy results have been obtained for some special cases, as described below.

A significant progress has been made in the {\bf crisp} case, i.e.\ when $\Gamma$ only contains functions $f:D^m\rightarrow \{0,+\infty\}$.
The problem is then called {\em Constraint Satisfaction} (CSP). The dichotomy is known to hold 
for languages with a 2-element domain (Schaefer~\cite{Schaefer78:complexity}),
languages with a 3-element domain (Bulatov~\cite{Bulatov06:3-elementJACM}),
conservative languages\footnote{A crisp language $\Gamma$ is called conservative if it contains all unary cost functions $f:D\rightarrow\{0,+\infty\}$ \cite{Bulatov03:conservative}.
A general-valued language is called conservative if it contains all unary cost functions $f:D\rightarrow \mathbb R_+$ \cite{KZ10:TRa,KZ10:TRb,K10:TRc}.
} (Bulatov \cite{Bulatov03:conservative}),
and languages containing a single relation without sources and sinks (Barto~\etal\ \cite{Barto09:siam}).
All dichotomy theorems above have the following form: if all functions in $\Gamma$ satisfy a certain
condition given by one or more {\em polymorphisms} then the language is tractable, otherwise it is NP-hard.

For general VCSPs the dichotomy has been shown to hold for Boolean languages, i.e.\ languages with a 2-element domain (Cohen~\etal\ \cite{Cohen:AI06}),
conservative languages (Kolmogorov and \v{Z}ivn\'y~\cite{KZ10:TRa,KZ10:TRb,K10:TRc}, who generalized
previous results by Deineko~\etal~\cite{DeinekoJKK08} and Takhanov \cite{Takhanov10:stacs}),
and $\{0,1\}$-valued languages with a 4-element domain (Jonsson~\etal~\cite{Jonsson:TR11}).
In these examples tractable subclasses are characterized by one or more {\em multimorphisms},
which are generalizations of polymorphisms. A multimorphism of arity $k$ over $D$ is
a tuple $\langle{\tt OP}_1,\ldots,{\tt OP}_k\rangle$ where ${\tt OP}_i$ is an operation $D^k\rightarrow D$.
Language $\Gamma$ is said to admit multimorphism $\langle{\tt OP}_1,\ldots,{\tt OP}_k\rangle$ if every function $f\in\Gamma$ satisfies
$$
f(\bx_1)+\ldots+f(\bx_k)
\ge
f({\tt OP}_1(\bx_1,\ldots,\bx_k))+\ldots+f({\tt OP}_k(\bx_1,\ldots,\bx_k))
$$
for all labelings $\bx_1,\ldots,\bx_k$ with $f(\bx_1)<+\infty$, $\ldots$, $f(\bx_k)<+\infty$.
(The pair of operations $\langle\sqcap,\sqcup\rangle$ used in~\eqref{eq:submodularity} is an example of a binary multimorphism.)
The tractable classes mentioned above (for $|D|>2$) are characterized by 
{\em complementary pairs of STP and MJN} multimorphisms~\cite{KZ10:TRb}
(that generalized {\em symmetric tournament pair (STP)} multimorphisms~\cite{Cohen:TCS08}),
and {\em 1-defect chain} multimorphisms~\cite{Jonsson:TR11}
(that generalized tractable weak-tree submodular functions in section \ref{sec:weak} originally introduced in~\cite{K10:TSv2}).

To make further progress on classifying complexity of VCSPs, it is important to study which
multimorphisms lead to tractable optimisation problems. Operations $\langle\sqcap,\sqcup\rangle$ and $\langle\wedge,\vee\rangle$ introduced in this paper
represent new classes of such multimorphisms: to our knowledge, previously researchers have not considered multimorphisms
defined on trees.

\myparagraph{Combining multimorphisms} Finally, we mention that some constructions, namely {\em Cartesian products} and {\em Malt'stev products},
can be used for obtaining new tractable classes of binary multimoprhisms from existing ones~\cite{KrokhinLarose:08}.
Note, Krokhin and Larose~\cite{KrokhinLarose:08} formulated these constructions only
for lattice multimorphisms $\langle\sqcap,\sqcup\rangle$, but the proof in~\cite{KrokhinLarose:08} actually applies to arbitrary binary multimorphisms $\langle\sqcap,\sqcup\rangle$.


\section{Steepest descent algorithm}\label{sec:alg}

It is known that for $L^\natural$-convex functions local optimality implies global optimality~\cite{Murota:book}.
We start by generalizing this result to strongly tree-submodular functions. 
Let us define the following ``local'' neighborhoods of labeling $\bx\in \D$:
\begin{eqnarray*}
\NEIB(\bx)    & = & \{\by\in \D\:|\: \rho(\bx,\by)\le 1 \} \\
\INWARD(\bx)  & = & \{\by\in \NEIB(\bx) \:|\: \by \preceq \bx \} \\
\OUTWARD(\bx) & = & \{\by\in \NEIB(\bx) \:|\:  \by \succeq \bx \} \\
\end{eqnarray*}
where $\bu\preceq\bv$ means that $u_i\preceq v_i$ for all $i\in V$, and $\rho(\bx,\by)=\max_{i\in V}\rho(x_i,y_i)$ is the $l_\infty$-distance between $\bx$ and $\by$.
Clearly, the restriction of $f$ to $\INWARD(\bx)$ is a submodular function,
and the restriction of $f$ to $\OUTWARD(\bx)$ is bisubmodular assuming that each tree $T_i$ is binary\footnote{If label $x_i$
has less than two children in $T_i$ then variable's domain after restriction will be a strict subset of $\{-1,0,+1\}$.
Therefore, we may need to use a bisubmodular minimization algorithm over a signed ring familiy $\calR\subseteq\{-1,0,+1\}^n$~\cite{McCormick:10}.}.

\begin{proposition} Suppose that {\em $f(\bx)=\min\{ f(\by) \: | \: \by\in\INWARD(\bx)\}  =  \min\{ f(\by) \: | \: \by\in\OUTWARD(\bx)\}$.} Then $\bx$ is a global minimum of $f$.
\label{prop:local}
\end{proposition}
\begin{proof}
First, let us prove that $f(\bx)=\min\{ f(\by) \: | \: \by\in\NEIB(\bx)\}$.
Let $\bx^\ast$ be a minimizer of $f$ in $\NEIB(\bx)$, and denote $\D^\ast=\{\by\in \D\:|\:y_i\in \D^\ast_i=\{x_i,x^\ast_i\}\}\subseteq\NEIB(\bx)$.
We treat set $\D^\ast_i$ as a tree with root $x_i\sqcap x^\ast_i$.
Clearly, the restriction of $f$ to $\D^\ast$ is an $L^\natural$-convex function under 
the induced operations $\sqcap$, $\sqcup$. It is known that for $L^\natural$-convex functions
optimality of $\bx$ in sets $\{\by\in \D^\ast\:|\: \by\preceq \bx\}$ and $\{\by\in \D^\ast\:|\: \by\succeq \bx\}$
suffices for optimality of $\bx$ in $\D^\ast$~\cite[Theorem 7.14]{Murota:book}, therefore $f(\bx)\le f(\bx^\ast)$.
This proves that $f(\bx)=\min\{ f(\by) \: | \: \by\in\NEIB(\bx)\}$.

Let us now prove that $\bx$ is optimal in $\D$.
Suppose not, then there exists $\by\in \D$ with $f(\by)<f(\bx)$.
Among such labelings, let us choose $\by$ with the minimum distance $\rho(\bx,\by)$. We must have $\by\notin\NEIB(\bx)$, so $\rho(\bx,\by)\ge 2$. Clearly,
$\rho(\bx,\bx\sqcup\by)\le \rho(\bx,\by)-1$ and $\rho(\bx,\bx\sqcap\by)\le \rho(\bx,\by)-1$.
Strong tree-submodularity and the fact that $f(\by)<f(\bx)$
imply that the cost of at least one of the labelings $\bx\sqcup\by$, $\bx\sqcap\by$ is smaller than $f(\bx)$.
This contradicts to the choice of $\by$.
\end{proof}

Suppose that each tree $T_i$ is binary.
The proposition shows that a greedy technique for computing a minimizer of $f$ would work. We can start with an arbitrary labeling $\bx\in \D$,
and then apply iteratively the following two steps in some order:
\begin{itemize}
\item[(1)] Compute minimizer $\bx^{\tt in}\in\arg\min\{f(\by)\:|\:\by\in\INWARD(\bx)\}$ by invoking a submodular minimization algorithm, replace $\bx$ with $\bx^{\tt in}$ if $f(\bx^{\tt in})<f(\bx)$.
\item[(2)] Compute minimizer $\bx^{\tt out}\in\arg\min\{f(\by)\:|\:\by \in\OUTWARD(\bx)\}$ by invoking a bisubmodular minimization algorithm, replace $\bx$ with $\bx^{\tt out}$ if $f(\bx^{\tt out})<f(\bx)$. 
\end{itemize}
The algorithm stops if neither step can decrease the cost. Clearly, it terminates in a finite number of steps and produces an optimal solution.
We will now discuss how to obtain a polynomial number of steps. We denote $K=\max_i |D_i|$.

\subsection{$L^\natural$-convex case}
For $L^\natural$-convex functions the {\em steepest descent} algorithm described above was first proposed by Murota~\cite{Murota:IEICE00,Murota:book,Murota:SIAM03},
except that in step 2 a submodular minimization algorithm was used.
Murota's algorithm actually computes both of $\bx^{\tt in}$ and $\bx^{\tt out}$ for the same $\bx$ and then chooses a better one by comparing costs
$f(\bx^{\tt in})$ and $f(\bx^{\tt out})$. A slight variation was proposed by Kolmogorov and Shioura~\cite{KS:09}, who allowed an arbitrary order of steps.
Kolmogorov and Shioura also established a tight bound on the number of steps of the algorithm by proving the following theorem.

\begin{theorem}[\cite{KS:09}]
Suppose that each tree $T_i$ is a chain. For a labeling $\bx\in\D$ define
\begin{subequations}
\begin{eqnarray}
\ifICALP
~\!\!\!\!\!\!\!\!\rho^-(\bx)\!=\!\min \{ \rho(\bx,\by)  |  \by \!\in\! OPT^{-}[\bx] \}, OPT^-[\bx]\!=\!\arg\!\min \{f(\by) |\by\!\in\!\D,\by\preceq\bx\} \; ~ \label{eq:rho:a} \\
~\!\!\!\!\!\!\!\!\rho^+(\bx)\!=\!\min \{ \rho(\bx,\by)  |  \by \!\in\! OPT^{+}[\bx] \}, OPT^+[\bx]\!=\!\arg\!\min \{f(\by) |\by\!\in\!\D,\by\succeq\bx\} \; ~ \label{eq:rho:b}
\else
\rho^-(\bx)=\min \{ \rho(\bx,\by) \: | \: \by \in OPT^{-}[\bx] \}, \quad OPT^-[\bx]=\arg \min \{f(\by) \:|\:\by\in\D,\by\preceq\bx\} \label{eq:rho:a} \\
\rho^+(\bx)=\min \{ \rho(\bx,\by) \: | \: \by \in OPT^{+}[\bx] \}, \quad OPT^+[\bx]=\arg \min \{f(\by) \:|\:\by\in\D,\by\succeq\bx\} \label{eq:rho:b}
\fi
\end{eqnarray}
\label{eq:rho}\vspace{-15pt}
\end{subequations}
\begin{itemize}
\item[(a)] Applying step (1) or (2) to labeling $\bx\!\in\!\D$ does not increase $\rho^-(\bx)$ and $\rho^+(\bx)$.
\item[(b)] If $\rho^-(\bx)\ge 1$ then applying step (1) to $\bx$ will decrease $\rho^-(\bx)$ by 1.
\item[(c)] If $\rho^+(\bx)\ge 1$ then applying step (2) to $\bx$ will decrease $\rho^+(\bx)$ by 1.
\end{itemize}
\label{th:KS}
\end{theorem}
In the beginning of the algorithm we have $\rho^-(\bx)\le K$ and $\rho^+(\bx)\le K$,
so the theorem implies that after at most $K$ calls to step (1) and $K$ calls to step (2) we get
$\rho^-(\bx)=\rho^+(\bx)=0$. The latter condition means that
$f(\bx)=\min\{ f(\by) \: | \: \by\in\INWARD(\bx)\}  =  \min\{ f(\by) \: | \: \by\in\OUTWARD(\bx)\}$,
and thus, by proposition~\ref{prop:local}, $\bx$ is a global minimum of $f$.

\subsection{General case}
We now show that the bound $O(K)$ on the number of steps 
is also achievable for general strongly tree-submodular functions. 
We will establish it for the following version of the steepest descent algorithm:
\begin{itemize}
\item[S0] Choose an arbitrary labeling $\bx^\circ\in \D$ and set $\bx:=\bx^\circ$.
\item[S1] Compute minimizer $\bx^{\tt in}\in\arg\min \{ f(\by) \:|\: \by\in\INWARD(\bx) \}$. If $f(\bx^{\tt in})<f(\bx)$
then set $\bx:=\bx^{\tt in}$ and repeat step S1, otherwise go to step S2.
\item[S2] Compute minimizer $\bx^{\tt out}\in\arg\min \{ f(\by) \:|\: \by\in\OUTWARD(\bx) \}$. If $f(\bx^{\tt out})<f(\bx)$ then set $\bx:=\bx^{\tt out}$ and repeat step S2, otherwise terminate.
\end{itemize}
Note, one could choose $x^\circ_i$ to be the root of tree $T_i$ for each node $i\in V$, then step {\tt S1} would be redundant. 

\begin{theorem}
(a) Step {\tt S1} is performed at most $K$ times. (b) Each step {\tt S2} preserves the following
property:
\begin{equation}
\mbox{\em $f(\bx)=\min \{f(\by)\:|\:\by\in\INWARD(\bx)\}$}
\label{eq:complexity:b}
\end{equation}
(c) Step {\tt S2} is performed at most $K$ times.
(d) Labeling $\bx$ produced upon termination of the algorithm is a minimizer of $f$.
\label{th:complexity}
\end{theorem}



\begin{proof}
For a labeling $\bx\in\D$ denote $\D^-[\bx]=\{\by\in\D\:|\:\by\preceq\bx\}$.
We will treat domain $\D^-[\bx]$ as the collection of chains with roots $r_i$ and leaves $x_i$. 
Let $\rho^-(\bx)$ be the quantity defined in~\eqref{eq:rho:a}. There holds
\begin{equation}
f(\bx)=\min \{f(\by)\:|\:\by\in\INWARD(\bx)\}\qquad \Leftrightarrow \qquad \rho^-(\bx)=0
\label{eq:StwoInv}
\end{equation}
Indeed, this equivalence can be obtained by applying proposition~\ref{prop:local} to function $f$ restricted to $\D^-[\bx]$.

\myparagraph{(a)}
When analyzing the first stage of the algorithm, we can assume without loss of generality
that $\D=\D^-[\bx^\circ]$, i.e.\ each tree $T_i$ is a chain with the root $r_i$ and
the leaf $x^\circ_i$. 
Indeed, removing the rest of the tree will not affect the behaviour of steps {\tt S1}.
With such assumption, function $f$ becomes $L^\natural$-convex. 
By theorem~\ref{th:KS}(b), steps {\tt S1} will terminate after at most $K$ steps.

\myparagraph{(b,c)} Property~\eqref{eq:complexity:b} (or equivalently $\rho^-(\bx)=0$) clearly
holds after termination of steps {\tt S1}. Let $\bz$ be the labeling upon
termination of steps {\tt S2}. 
When analyzing the second stage of the algorithm, we can assume without loss of generality
that $\D=\D^-[\bz]$, i.e.\ each tree $T_i$ is a chain with the root $r_i$ and
the leaf $z_i$. 
Indeed, removing the rest of the tree will not affect the behaviour of steps {\tt S2}.
Furthermore, restricting $f$ to $\D^-[\bz]$ does not affect the definition of $\rho^-(\bx)$ for $\bx\in\D^-[\bz]$.

By theorem~\ref{th:KS}(a), steps {\tt S2} preserve $\rho^-(\bx)=0$; this proves part (b).
Part (c) follows from theorem~\ref{th:KS}(c).

\myparagraph{(d)} When steps {\tt S2} terminate, we have $f(\bx)=\min \{f(\by)\:|\:\by\in\OUTWARD(\bx)\}$.
Combining this fact with condition~\eqref{eq:complexity:b} and using proposition~\ref{prop:local}
gives that upon algorithm's termination $\bx$ is a minimizer of $f$.

\end{proof}



\section{Translation submodularity}\label{sec:sw}
In this section we derive an alternative definition of strongly tree-submodular functions.
As a corollary, we will obtain that strong tree submodularity~\eqref{eq:submodularity}
implies weak tree submodularity~\eqref{eq:submodularity:weak}.

Let us introduce another pair of operations on trees. 
Given labels $a,b\in D_i$ and an integer $d\ge 0$, we define 
$$
a \uparrow^d b   = \calP[a\rightarrow b,d]\wedge b \qquad
a \downarrow_d b = \calP[a\rightarrow b,\rho(a \uparrow^d b,b)]
$$
In words, $a \uparrow^d b$ is obtained as follows: (1) move from $a$ towards $b$ by $d$ steps,
stopping if $b$ is reached earlier; (2) keep moving until the current label becomes an ancestor of $b$.
$a \downarrow_d b$ is the label on the path $\calP[a\rightarrow b]$ such that the distances
$\rho(a,a \downarrow_d b)$ and $\rho(a \uparrow^d b,b)$ are the same,
as well as distances $\rho(a,a \uparrow^d b)$ and $\rho(a \downarrow_d b,b)$.
Note, binary operations $\uparrow^d,\downarrow_d:D_i\times D_i\rightarrow D_i$ 
(and corresponding operations $\uparrow^d,\downarrow_d:\D\times \D\rightarrow \D$) are in general non-commutative.
One exception is $d=0$, in which case  $\uparrow^d,\downarrow_d$ reduce to the commutative
operations defined in the introduction: $\bx\uparrow^0\by=\bx\wedge\by$ and $\bx\downarrow_0\by=\bx\vee\by$.

For fixed labels $a,b\in D_i$ it will often be convenient to rename nodes in $\calP[a\rightarrow b]$ to be
consecutive integers so that $a\wedge b=0$ and $a\le 0\le b$. Then we have $a=-\rho(a,a\wedge b)$, $b=\rho(a\wedge b,b)$ and
$$
a \uparrow^d b = \max\{0,\min\{ a+d,b \}\} \qquad
a \downarrow_d b = a + b - (a \uparrow^d b)
$$

\begin{theorem} (a) If $f$ is strongly tree-submodular then for any $\bx, \by\in \D$ and integer $d\ge 0$ there holds
{\em
\begin{equation}
f(\bx)+f(\by)\ge f(\bx \uparrow^d \by) + f(\bx \downarrow_d \by)
\label{eq:submodularity:dir}
\end{equation}
\em }
(b) If~\eqref{eq:submodularity:dir} holds for any $\bx, \by\in \D$ and $d\ge 0$ then $f$ is strongly tree-submodular.
\label{th:submodularity:dir}
\end{theorem}

Note, this result is well-known for $L^\natural$-convex functions~\cite[section 7.1]{Murota:book},
i.e.\ when all trees are chains shown in Figure~\ref{fig:trees}(b); inequality~\eqref{eq:submodularity:dir} was then written as 
$f(\bx)+f(\by)\ge f((\bx+d \cdot{\bf 1})\wedge \by) + f(\bx\vee (\by- d \cdot{\bf 1}))$, and was called {\em translation submodularity}.
In fact, translation submodularity is one of the key properties of $L^\natural$-convex functions, and was heavily used, for example, in~\cite{KS:09}
for proving theorem~\ref{th:KS}.

Setting $d=0$ in theorem \ref{th:submodularity:dir}(a) gives
\begin{corollary}
A strongly tree-submodular function $f$ is also weakly tree-submodular, i.e.\ \eqref{eq:submodularity} implies \eqref{eq:submodularity:weak}.
\end{corollary}

A proof of parts (b) and (a) of theorem~\ref{eq:submodularity:dir} is given 
in sections \ref{ref:proof:partb} and \ref{ref:proof:parta} respectively.
In both proofs we always implicitly assume
that for each $i\in V$ labels in $\calP[x_i\rightarrow y_i]$ are renamed to be consecutive integers with $x_i\wedge y_i=0$ and $x_i\le 0\le y_i$.

\subsection{Proof of theorem~\ref{th:submodularity:dir}(b)}\label{ref:proof:partb}
We prove inequality \eqref{eq:submodularity} for $\bx,\by\in\D$ using induction on $\rho_1(\bx,\by)=\sum_{i\in V}\rho(x_i,y_i)$.
The base case $\rho_1(\bx,\by)=0$, or $\bx=\by$, is trivial; suppose that $\rho_1(\bx,\by)\ge 1$.
Denote $d_{\max}=\rho(\bx,\by)\ge 1$ and $d=\lfloor d_{\max}/2 \rfloor\ge 0$. Two cases are possible.

\myparagraph{\underline{Case 1}} $d_{\max}$ is even.
We can assume without loss of generality that there exists $k\in V$ such that $y_k-x_k=d_{\max}$ and
$|x_k|\ge y_k$.
(If there is no such $k$, we can simply swap $\bx$ and $\by$; inequality \eqref{eq:submodularity}
will be unaffected since operations $\sqcap,\sqcup$ are commutative, and $\rho(\bx,\by)$, $\rho_1(\bx,\by)$ will not change.)
Consider labelings $\bx',\by'\in\D$ defined as follows:
\begin{equation*}
y'_i  =  \begin{cases}
y_i-1 & \mbox{ if~} y_i-x_i=d_{\max}, \; |x_i| \ge y_i \\
y_i   & \mbox{ otherwise}
\end{cases}
 \qquad\qquad\qquad
x'_i  =  x_i \sqcup y'_i
\end{equation*}
for each $i\in V$. We claim that
\begin{equation*}
\begin{tabular}{ll}
(a)$\quad \bx \sqcap \by' = \bx \sqcap \by \qquad \quad  $ &
(b)$\quad \bx \sqcup \by' = \bx' $ \\
(c)$\quad \bx' \uparrow^{d} \by =  \by' $ &
(d)$\quad \bx' \downarrow_{d} \by = \bx \sqcup \by$
\end{tabular}
\end{equation*}
Indeed, for each node $i\in V$ one of the following holds:
\begin{itemize}
\item[$\bullet$] $y_i-x_i\le d_{\max}-1$. Then $y'_i=y_i$, $x'_i=x_i \sqcup y_i$, so (a) and (b) hold for node $i$.
We also have $y_i - x'_i = y_i - (x_i\sqcup y_i) \le \lceil (y_i - x_i)/2 \rceil \le \lceil (d_{\max}-1)/2 \rceil \le d$,
which implies (c) and (d). 
\item[$\bullet$] $y_i-x_i=d_{\max}$ and $|x_i|<y_i$. Then $y'_i=y_i$, $x'_i=x_i \sqcup y_i=(x_i+y_i)/2$, 
$y_i - x'_i = d$; as above, this implies (a)-(d). 
\item[$\bullet$] $y_i-x_i=d_{\max}$ and $|x_i|\ge y_i$. Then $y'_i=y_i-1$, $x'_i = x_i \sqcup y'_i = \lfloor(x_i+y_i-1)/2\rfloor=(x_i+y_i)/2-1=y'_i-d$.
Checking that (a)-(d) hold is straightforward.
\end{itemize}
We have $y'_k=y_k-1$, and so $\rho_1(\bx,\by')<\rho_1(\bx,\by)$. Therefore, 
\begin{eqnarray*}
f(\bx)+f(\by') \ge  f(\bx \sqcap \by') + f(\bx \sqcup \by') & \quad &
f(\bx')+f(\by) \ge  f(\bx' \uparrow^{d} \by) + f(\bx' \downarrow_{d} \by)
\end{eqnarray*}
where the first inequality follows from the induction hypothesis and the second one follows from \eqref{eq:submodularity:dir}.
Summing these inequalities and subtracting $f(\bx')+f(\by')$ from both sides using (a)-(d) gives \eqref{eq:submodularity}.

\myparagraph{\underline{Case 2}} $d_{\max}$ is odd.
By swapping $\bx$ and $\by$, if necessary, we can assume without loss of generality that there exists $k\in V$ such that $y_k-x_k=d_{\max}$ and
$|x_k|<y_k$. (Note, we cannot have $y_i-x_i=d_{\max}$ and $|x_i|=y_i$ since $d_{\max}$ is odd).
Consider labelings $\bx',\by'\in\D$ defined as follows:
\begin{equation*}
x'_i  =  \begin{cases}
x_i+1 & \mbox{ if~} y_i-x_i=d_{\max}, \; |x_i| < y_i \\
x_i   & \mbox{ otherwise}
\end{cases}
 \qquad\qquad\qquad
y'_i  =  x'_i \sqcap y_i
\end{equation*}
for each $i\in V$. We claim that
\begin{equation*}
\begin{tabular}{ll}
(a)$\quad \bx' \sqcap \by = \by' \qquad \quad  $ &
(b)$\quad \bx' \sqcup \by = \bx \sqcup \by $ \\
(c)$\quad \bx \uparrow^{d} \by' =  \bx \sqcap \by $ &
(d)$\quad \bx \downarrow_{d} \by' = \bx'$
\end{tabular}
\end{equation*}
Indeed, for each node $i\in V$ one of the following holds:
\begin{itemize}
\item[$\bullet$] $y_i-x_i\le d_{\max}-1$. Then $x'_i=x_i$, $y'_i=x_i \sqcap y_i$, so (a) and (b) hold for node $i$.
We also have $y'_i - x_i = (x_i\sqcap y_i)-x_i \le \lceil (y_i - x_i)/2 \rceil \le \lceil (d_{\max}-1)/2 \rceil \le d$,
which implies (c) and (d).
\item[$\bullet$] $y_i-x_i=d_{\max}$ and $|x_i|>y_i$. Then $x'_i=x_i$, $y'_i=x_i \sqcup y_i=\lceil(x_i+y_i)/2\rceil\le 0$,
so (a) and (b) hold for node $i$.
(c) and (d) hold since $y'_i\le 0$.
\item[$\bullet$] $y_i-x_i=d_{\max}$ and $|x_i|<y_i$. Then $x'_i=x_i+1$,
$y'_i = x'_i \sqcap y_i = \lfloor(x_i+y_i-1)/2\rfloor$.
Checking that (a)-(d) hold is straightforward.
\end{itemize}
We have $x'_k=x_k+1$, and so $\rho_1(\bx',\by)<\rho_1(\bx,\by)$. Therefore, 
\begin{eqnarray*}
f(\bx')+f(\by) \ge  f(\bx' \sqcap \by) + f(\bx' \sqcup \by) & \quad &
f(\bx)+f(\by') \ge  f(\bx \uparrow^{d} \by') + f(\bx \downarrow_{d} \by')
\end{eqnarray*}
where the first inequality follows from the induction hypothesis and the second one follows from \eqref{eq:submodularity:dir}.
Summing these inequalities and subtracting $f(\bx')+f(\by')$ from both sides using (a)-(d) gives \eqref{eq:submodularity}.

\subsection{Proof of theorem~\ref{th:submodularity:dir}(a)}\label{ref:proof:parta}

We say that the triplet $(\bx,\by,d)$ is {\em valid} if $\bx,\by\in \D$ and $d\in [0,\rho(\bx,\by)]$.
We denote $\bz=\bx\uparrow^d \by$; we have $x_i\le 0\le z_i\le y_i$.
Let us introduce a partial order $\preceq$ over valid triplets as the lexicographical order with variables
$(y_1-x_1,\ldots,y_n-x_n,-d)$. Note, the last component $-d$ is the least significant.
We use induction on this partial order. 
The induction base is trivial: if the first $n$ components are zeros then $\bx=\by$ so~\eqref{eq:submodularity:dir} is an equality,
and if the last component is minimal (i.e.\ $d=\rho(\bx,\by)$) then $\bx\uparrow^d\by=\by$ and $\bx\downarrow_d\by=\bx$,
so~\eqref{eq:submodularity:dir} is again an equality. Suppose that $\bx\ne \by$ and $d\le \rho(\bx,\by)-1$. 

Consider integer $d'\ge d$, and denote $\by'=\bx\uparrow^{d'+1}\by$ and $\delta_i=y'_i-z_i\ge 0$ for $i\in V$.
Suppose that $\delta_i \in \{0,1\}$ for all nodes $i\in V$. (This holds, for example, if $d'=d$.)
Denote $\bx'=\bx \downarrow_d \by'$, then $x'_i=x_i+y'_i- (x_i \uparrow^d y'_i)=x_i+y'_i-z_i=x_i+\delta_i$.
We claim that
\begin{equation}
\begin{tabular}{ll}
(a)$\quad \bx \uparrow^d \by' = \bx \uparrow^d \by \qquad \quad  $ &
(b)$\quad \bx \downarrow_d \by' = \bx' $ \\
(c)$\quad \bx' \uparrow^{d'} \by =  \by' $ &
(d)$\quad \bx' \downarrow_{d'} \by = \bx \downarrow_d \by$
\end{tabular}
\label{eq:dir:proof}
\end{equation}
In order to prove it, let us consider node $i$. Property (a) follows from the fact that $y'_i \ge x_i\uparrow^d y_i$.
Property (b) is the definition of $\bx'$. To prove (c), consider two possible cases:
\begin{itemize}
\item $\delta_i=0$, so $x'_i=x_i$ and $y'_i \equiv x_i \uparrow^{d'+1} y_i = x_i\uparrow^d y_i$. The latter condition and the fact
$d'+1>d$ imply that $x_i+d\ge y_i$, therefore $x'_i+d' \ge y_i = y'_i$. This leads to (c).
\item $\delta_i=1$. If $y'_i=y_i$ then condition (c) is straightforward (it follows from $x'_i\uparrow^{d'} y_i\ge y'_i$). 
Suppose that $y'_i < y_i$, then from definition of $y'_i$ we have $x_i+d'+1\le y'_i$, or $x'_i+d'\le y'_i$. This leads to (c).
\end{itemize}
Finally, properties (c) and (d) are equivalent since
\ifICALP
\begin{eqnarray*}
x'_i + y_i - y'_i - [x_i \downarrow_d y_i] 
& = & [x_i + y'_i - (x_i\uparrow^d y'_i)] + y_i - y'_i - [x_i + y_i - (x_i \uparrow^d y_i)] \\
& = & (x_i\uparrow^d y_i) - (x_i\uparrow^d y'_i)
= 0
\end{eqnarray*}
\else
$$
x'_i + y_i - y'_i - [x_i \downarrow_d y_i] 
= [x_i + y'_i - (x_i\uparrow^d y'_i)] + y_i - y'_i - [x_i + y_i - (x_i \uparrow^d y_i)] 
= (x_i\uparrow^d y_i) - (x_i\uparrow^d y'_i)
= 0
$$
\fi

Now suppose that in addition to conditions $\delta_i\in\{0,1\}$ there holds $\bx'\ne \bx$ and $\by'\ne \by$. 
Then we have $(\bx,\by',d)\prec (\bx,\by,d)$ and $(\bx',\by,d')\prec (\bx,\by,d)$, so by the induction hypothesis
\begin{eqnarray*}
f(\bx)\!+\!f(\by') \ge  f(\bx \uparrow^d \by')\!+\!f(\bx \downarrow_d \by') & \quad &
f(\bx')\!+\!f(\by) \ge  f(\bx' \uparrow^{d'} \by)\!+\!f(\bx' \downarrow_{d'} \by)
\end{eqnarray*}
Summing these inequalities and subtracting $f(\bx')+f(\by')$ from both sides using (a)-(d) gives \eqref{eq:submodularity:dir}.

Let us describe cases when the argument above can be applied; such cases can be eliminated from consideration. First, suppose that $y_j-z_j\ge 2$ for some node $j\in V$,
then there exists $d'\ge d$ such that the labeling $\by'=\bx\uparrow^{d'+1}\by$ has at least one node $j\in V$
with $y'_j\in[z_j+1,y_j-1]$. Let us choose the minimum integer $d'$ that has this property. Then $\delta_i\in\{0,1\}$ for all
nodes $i\in V$, since $\delta_i\ge 2$ would contradict to the minimality of chosen $d'$. We also have $y'_j\ne y_j$ and 
$x'_j\ne x_j$ (since $x'_j-x_j=\delta_j=1$), so the conditions above are satisfied. Therefore, from now on we assume without loss of generality
that $y_i-z_i\in\{0,1\}$ for all nodes $i\in V$.

We can also take $d'=d$. Condition $\delta_i\in\{0,1\}$ is then satisfied for all nodes. Therefore, we can assume without loss of generality 
that either $\bx'=\bx$ or $\by'=\by$ where $\by'=\bx\uparrow^{d+1}\by$, $\bx'= \bx\downarrow_d\by'$, otherwise the induction argument above could be applied.
Suppose that $\bx'=\bx$. This is equivalent to $\bx\uparrow^d \by'=\by'$,
or to the following condition for all nodes $i\in V$: either $x_i+d<0$ or $y_i-x_i\ge d$. It can be checked that
$\bx \uparrow^{d+1}\by= \bx \uparrow^{d}\by$ and $\bx \downarrow_{d+1}\by = \bx \downarrow_{d}\by$. Furthermore,
$(\bx,\by,d+1)\prec(\bx,\by,d)$, so~\eqref{eq:submodularity:dir} follows by the induction hypothesis. We thus assume from now on
that $\by'=\by$.

Equations below summarize definitions and assumptions made so far:
\begin{subequations}
\begin{eqnarray}
z_i &=& x_i \uparrow^d y_i \\
y'_i &=& x_i \uparrow^{d+1} y_i = y_i \\
x'_i & = & x_i \downarrow_d y_i \\
\delta_i & = & y_i - z_i = x'_i - x_i \in \{0,1\}
\end{eqnarray}

Let $S$ be the set of nodes $i\in V$ with $\delta_i=1$. It is straighforward to check that 
\begin{eqnarray}
i \in S\hspace{22pt}    & \Rightarrow & x_i + d = z_i = y_i-1 \\
i \in V-S & \quad\Rightarrow\quad & x_i \uparrow^d y_i = y_i \mbox{~~~and~~~} x_i\downarrow_d y_i=x_i=x'_i
\end{eqnarray}
\end{subequations}

If $S$ is empty then $\bx\uparrow^d \by=\by$, $\bx\downarrow_d \by=\bx$, so inequality~\eqref{eq:submodularity:dir}
is trivial. Thus, we can assume that $S$ is non-empty.
Suppose that $S$ contains two distinct nodes $i$ and $j$. Let us modify labelings $\bx'$ and $\by'$ as follows:
for node $j$ set $x'_j=x_j$, $y'_j=z_j$. It is straightforward to check that conditions~\eqref{eq:dir:proof} for $d'=d$
still hold. Furthermore, $x'_i>x_i$, $y'_j<y_j$, so $(\bx,\by',d)\prec (\bx,\by,d)$ and $(\bx',\by,d)\prec (\bx,\by,d)$.
Applying the argument described above gives~\eqref{eq:submodularity:dir}.

We are left with the case when $S$ contains a single node $j$. 
We will consider 5 possible subcases.
In 4 of them, we will do the following: (i) specify new labelings $\bx'$ and $\by'$ with $x'_i,y'_i\in[x_i,y_i]$ for each node $i$;
(ii) specify four identities involving $\bx,\by,\bx',\by'$ such that the right-hand sides contain expressions $\bx',\by',\bx\uparrow^d \by,\bx\downarrow_d\by$,
 and the left-hand sides contain expressions of the form 
$\bx \diamond_1 \by'$, $\bx ~\overline\diamond_1~ \by'$, 
$\bx' \diamond_2 \by$, $\bx' ~\overline\diamond_2~ \by$ 
where $\diamond_k$ is one of the operations $\sqcap,\sqcup,\uparrow^{d_k},\downarrow_{d_k}$
and $\overline\diamond_k$ is the corresponding ``symmetric'' operation. This will describe how to prove~\eqref{eq:submodularity:dir}:
we would need to sum two inequalities
$$
f(\bx)+f(\by')\ge f(\bx \diamond_1 \by')+f(\bx ~\overline\diamond_1~ \by') \qquad
f(\bx')+f(\by)\ge f(\bx \diamond_2 \by')+f(\bx ~\overline\diamond_2~ \by')
$$
that hold either by strong tree-submodularity or by the induction hypothesis, then use provided identities to prove~\eqref{eq:submodularity:dir}.
Checking the identities and the applicability of the induction hypothesis in the case of operations $\uparrow^{d_k}$, $\downarrow_{d_k}$
is mechanical, and we omit it.

\vspace{3pt}
\noindent \underline{\bf Case 1}~~$z_j=x_j+d\ge 1$ (implying $d \ge 1$). The identities are
\begin{equation}
\begin{tabular}{ll}
(a)$\quad \bx \uparrow^{d-1} \by' =  \bx' $ &
(b)$\quad \bx \downarrow_{d-1} \by' = \bx \downarrow_d \by$ \\
(c)$\quad \bx' \sqcup \by = \bx \uparrow^d \by \qquad \quad  $ &
(d)$\quad \bx' \sqcap \by = \by' $ 
\end{tabular}
\end{equation}
and labelings $\bx',\by'$ are defined as follows:
%
\begin{itemize}
\item if $i=j$ set $x'_j=y_j-2$, $y'_j=y_j-1$; 
\item otherwise if $x_i+d=y_i>0$ set $x'_i=y'_i=y_i-1$.
\item otherwise (if $y_i=0$ or $x_i+d>y_i$) set $x'_i=y'_i=y_i$.
\end{itemize}


The remainder is devoted to the case $z_j=x_j+d=0$. Note that we must have $y_j=1$.

\vspace{3pt}
\noindent \underline{\bf Case 2}~~$d\ge 1$, $z_j=x_j+d=0$ and there exists node $k\in V-\{j\}$ with $x_k=0$, $y_k>0$.  Then
\begin{equation}
\begin{tabular}{ll}
(a)$\quad \bx' \uparrow^{d} \by =  \bx \uparrow^d \by$ &
(b)$\quad \bx' \downarrow_{d} \by = \by' $ \\
(c)$\quad \bx \sqcup \by' =  \bx' $ &
(d)$\quad \bx \sqcap \by' = \bx \downarrow_d \by \qquad \quad  $ 
\end{tabular}
\end{equation}
$\bx'$, $\by'$ are defined as follows: 
\begin{itemize}
\item if $i=j$ set $x'_j=x_j$, $y'_j=x_j+1$; 
\item otherwise if $i=k$ set $x'_k=y'_k=x_k+1=1$;
\item otherwise set $x'_i=y'_i=x_i$.
\end{itemize}

\vspace{3pt}
\noindent \underline{\bf Case 3}~~$d\ge 1$, $z_j=x_j+d=0$ and there is no node $k\in V-\{j\}$ with $x_k=0$, $y_k>0$. 
The identities are
\begin{equation}
\begin{tabular}{ll}
(a)$\quad \bx' \uparrow^{d-1} \by =  \bx \uparrow^d \by$ &
(b)$\quad \bx' \downarrow_{d-1} \by = \by' $ \\
(c)$\quad \bx \sqcup \by' =  \bx \downarrow_d \by \qquad \quad$ &
(d)$\quad \bx \sqcap \by' =  \bx'  $ 
\label{eq:GHAOUHAD}
\end{tabular}
\end{equation}
$\bx'$, $\by'$ are defined as follows: 
\begin{itemize}
\item if $i=j$ set $x'_j=x_i+1$, $y'_j=x_i+2$; 
\item otherwise if $x_i<0$ set $x'_i=y'_i=x_i+1$;
\item otherwise (if $x_i=y_i=0$) set $x'_i=y'_i=0$.
\end{itemize}
Note, to verify identities~\eqref{eq:GHAOUHAD} for node $j$, one should consider cases $d=1$ and $d\ge 2$ separately.

\vspace{3pt}
\noindent \underline{\bf Case 4}~~$d=0$ (implying $x_j=0$, $y_j=1$) and there exists node $k\in V-\{j\}$
with $x_k<0$. Then
\begin{equation}
\begin{tabular}{ll}
(a)$\quad \bx \uparrow^0 \by' =  \bx' $ &
(b)$\quad \bx \downarrow_0 \by' = \bx \downarrow_0 \by \qquad \quad  $ \\
(c)$\quad \bx' \sqcup \by =  \by'$ &
(d)$\quad \bx' \sqcap \by = \bx\uparrow^0 \by $ 
\end{tabular}
\end{equation}
$\bx'$, $\by'$ are defined as follows: 
\begin{itemize}
\item if $i=j$ set $x'_j=0$, $y'_j=1$; 
\item otherwise if $x_i<0$, $y_i=0$ set $x'_i=y'_i=-1$;
\item otherwise (if $x_i=y_i=0$) set $x'_i=y'_i=0$.
\end{itemize}

\vspace{3pt}
\noindent \underline{\bf Case 5}~~$d=0$ (implying $x_j=0$, $y_j=1$) and there is no node $k\in V-\{j\}$ with $x_k<0$. 
Thus, $x_i=y_i=0$ for all $i\in V-\{j\}$. There holds $\bx\uparrow^0\by=\bx$, $\bx\downarrow_0\by=\by$,
so inequality~\eqref{eq:submodularity:dir} is trivial.




\section{Weakly tree-submodular functions}\label{sec:weak}
In this section we consider functions $f$ that satisfy condition~\eqref{eq:submodularity:weak},
but not necessarily condition~\eqref{eq:submodularity}.
It is well-known~\cite{Topkis:78,Murota:book} that such functions can be minimized efficiently 
if all trees $T_i$ are chains rooted at an endpoint and $\max_i |D_i|$ is polynomially bounded.
The algorithm utilizes Birkhoff's representation theorem~\cite{Birkhoff:37} which says that there exists a {\em ring family} $\calR$
such that there is an isomorphism between sets $\D$ and $\calR$ that preserves operations $\wedge$ and $\vee$.
(A subset $\calR\subseteq\{0,1\}^m$ is a ring family if it is closed under operations $\wedge$ and $\vee$.)
It is known that submodular functions over a ring family can be minimized in polynomial time, which implies the result.
Note that the number of variables will be $m=O(\sum_{i}|D_i|)$.

Another case when $f$ satisfying~\eqref{eq:submodularity:weak} can be minimized efficiently is when $f$ is bisubmodular, i.e.\ all trees
are as shown in Figure~\ref{fig:trees}(c). Indeed, in this case the pairs of operations $\langle\sqcap,\sqcup\rangle$ and $\langle\wedge,\vee\rangle$ coincide.

An interesting question is whether there exist other classes of weakly tree-submodular functions that
can be minimized efficiently. In this section we provide one rather special example. We consider the tree shown in
Figure~\ref{fig:trees}(d). Each $T_i$ has nodes $\{0,1,\ldots,K,K_{-1},K_{+1}\}$ such that $0$ is the root, the parent of
$k$ for $k=1,\ldots,K$ is $k-1$, and the parent of $K_{-1}$ and $K_{+1}$ is $K$. 

In order to minimize function $f$ for such choice of trees, we create $K+1$ variables $y_{i0},y_{i1},\ldots,y_{iK}$ for
each original variable $x_i\in D_i$. The domains of these variables are as follows:
$\tilde D_{i0}=\ldots=\tilde D_{iK-1}=\{0,1\}$, $\tilde D_{iK}=\{-1,0,+1\}$. Each domain 
is treated as a tree with root $0$ and other nodes being the children of $0$; this defines operations $\wedge$ and $\vee$
for domains $\tilde D_{i0},\ldots \tilde D_{iK-1}, \tilde D_{iK}$. The domain $\tilde \D$ is set as the Cartesian product of individual
domains over all nodes $i\in V$. Note, a vector $\by\in \tilde \D$ has $n(K+1)$ components.

For a labeling $\bx\in \D$ let us define labeling $\by=\psi(\bx)\in\tilde \D$ as follows:
\begin{eqnarray*}
x_i=k\in\{0,1,\ldots,K\} & \quad \Rightarrow \quad & y_{i0}=\ldots=y_{ik-1}=1 , \; y_{ik}=\ldots=y_{iK}=0 \\
x_i=K_{-1} \hspace{60pt} & \quad \Rightarrow \quad & y_{i0}=\ldots=y_{iK-1}=1 , \; y_{iK}=-1 \\
x_i=K_{+1}\hspace{60pt} & \quad \Rightarrow \quad & y_{i0}=\ldots=y_{iK-1}=1 , \; y_{iK}=+1 
\end{eqnarray*}
%
%
\ignore{
\begin{equation*}
\begin{tabular}{c|cccccc}
$x_i$ & $y_{i1}$ & $y_{i2}$ & $\ldots$ & $y_{iK-1}$ & $y_{iK}$ & $y_i$ \\
\hline
$0$ & $0$ & $0$ & $\ldots$ & $0$ & $0$ & $0$ \\
$1$ & $1$ & $0$ & $\ldots$ & $0$ & $0$ & $0$ \\
$2$ & $1$ & $1$ & $\ldots$ & $0$ & $0$ & $0$ \\
$\ldots$ & $\ldots$ & $\ldots$ & $\ldots$ & $\ldots$ & $\ldots$ & $\ldots$ \\
$K-1$ & $1$ & $1$ & $\ldots$ & $1$ & $0$ & $0$ \\
$K$ & $1$ & $1$ & $\ldots$ & $1$ & $1$ & $0$ \\
$K_{-1}$ & $1$ & $1$ & $\ldots$ & $1$ & $1$ & $-1$ \\
$K_{+1}$ & $1$ & $1$ & $\ldots$ & $1$ & $1$ & $+1$ \\
\end{tabular}
\end{equation*}
}
It is easy to check that mapping $\psi:\D\rightarrow \tilde \D$ is injective and preserves operations $\wedge$ and $\vee$.
Therefore, $\calR={\tt Im}~ \psi$ is a {\em signed ring family}, i.e.\ a subset of $\tilde \D$ closed under operations $\wedge$ and $\vee$.
It is known~\cite{McCormick:10} that bisubmodular functions over ring families can be minimized in polynomial time, leading to
\begin{proposition}
Functions that are weakly tree-submodular with respect to trees shown in Figure~\ref{fig:trees}(d) can be minimized
in time polynomial in $n$ and $\max_i |D_i|$.
\end{proposition}

\section{Conclusions and discussion}
We introduced two classes of functions (strongly tree-submodular and weakly tree-submodular)
that generalize several previously studied classes. For each class, we gave new examples of trees
for which the minimization problem is tractable.

Our work leaves a natural open question: what is the complexity of the problem for more general trees?
In particular, can we minimize efficiently strongly tree-submodular functions if trees are non-binary,
i.e.\ if some nodes have three or more children? Note that the algorithm in section~\ref{sec:alg}
and its analysis are still valid, but it is not clear whether the minimization procedure in step {\tt S2}
can be implemented efficiently. Also, are there trees besides the one shown in Figure~\ref{fig:trees}(d)
for which weakly tree-submodular functions can be minimized efficiently?

More generally, can one characterize for which operations $\langle\sqcap,\sqcup\rangle$ the minimization problem
is tractable? Currently known tractable examples are distributive lattices, some non-distributive lattices~\cite{KrokhinLarose:08,Kuivinen:TR},
operations on trees introduced in this paper, and combinations of the above operations obtained via Cartesian product and Malt'sev product~\cite{KrokhinLarose:08}.
Are there tractable cases that cannot be obtained via lattice and tree-based operations?

\small
\bibliographystyle{plain}

\begin{thebibliography}{10}

\bibitem{Barto09:siam}
L.~Barto, M.~Kozik and T.~Niven.
\newblock The {C}{S}{P} dichotomy holds for digraphs with
  no sources and no sinks (a positive answer to a conjecture of {B}ang-{J}ensen
  and {H}ell).
\newblock {\em {SIAM} Journal on Computing}, 38(5):1782--1802, 2009.

\bibitem{Birkhoff:37}
Garrett Birkhoff.
\newblock Rings of sets.
\newblock {\em Duke Mathematical Journal}, 3(3):443--454, 1937.

\bibitem{Bouchet:87}
A.~Bouchet.
\newblock Greedy algorithm and symmetric matroids.
\newblock {\em Math. Programming}, 38:147--159, 1987.

\bibitem{Bouchet:95}
A.~Bouchet and W.~H. Cunningham.
\newblock Delta-matroids, jump systems and bisubmodular polyhedra.
\newblock {\em SIAM J. Discrete Math.}, 8:17--32, 1995.

\bibitem{Bulatov03:conservative}
A.~A. Bulatov.
\newblock Tractable {C}onservative {C}onstraint {S}atisfaction {P}roblems.
\newblock In {\em Proceedings of the 18th {I}{E}{E}{E} {S}ymposium on {L}ogic in {C}omputer
  {S}cience ({L}{I}{C}{S}'03)}, pages 321--330, 2003.

\bibitem{Bulatov06:3-elementJACM}
A.~A. Bulatov.
\newblock A dichotomy theorem for constraint satisfaction problems on a
  3-element set
\newblock {\em Journal of the ACM}, 53(1):66--120, 2006.

\bibitem{Chandrasekaran:88}
R.~Chandrasekaran and Santosh~N. Kabadi.
\newblock Pseudomatroids.
\newblock {\em Discrete Math.}, 71:205--217, 1988.

\bibitem{Cohen:CP04}
David Cohen, Martin Cooper, and Peter Jeavons.
\newblock A complete characterization of complexity for boolean constraint
  optimization problems.
\newblock In {\em Principles and Practice of Constraint Programming}, number
  3258 in Lecture Notes in Computer Science, pages 212--226, 2004.

\bibitem{Cohen:CP03}
David Cohen, Martin Cooper, Peter Jeavons, and Andrei Krokhin.
\newblock Soft constraints: complexity and multimorphsims.
\newblock In {\em Principles and Practice of Constraint Programming}, number
  2833 in Lecture Notes in Computer Science, pages 244--258, 2003.

\bibitem{Cohen:TCS08}
David~A. Cohen, Martin~C. Cooper, and Peter~G. Jeavons.
\newblock Generalising submodularity and horn clauses: Tractable optimization
  problems defined by tournament pair multimorphisms.
\newblock {\em Theoretical Computer Science}, 401:36--51, 2008.

\bibitem{Cohen:AI06}
David~A. Cohen, Martin~C. Cooper, Peter~G. Jeavons, and Andrei~A. Krokhin.
\newblock The complexity of soft constraint satisfaction.
\newblock {\em Artificial Intelligence}, 170:983--1016, 2006.

\bibitem{DeinekoJKK08}
V.~Deineko, P.~Jonsson, M.~Klasson and A.~Krokhin.
\newblock The approximability of {Max}
  {C{S}{P}} with fixed-value constraints.
\newblock {\em Journal of the ACM}, 55(4), 2008.

\bibitem{FavatiTardella:90}
P.~Favati and F.~Tardella.
\newblock Convexity in nonlinear integer programming.
\newblock {\em Ricerca Operativa}, 53:3--44, 1990.

\bibitem{Feder98:monotone}
T.~Feder and M.~Vardi.
\newblock The {C}omputational {S}tructure of {M}onotone {M}onadic
  {S{N}{P}} and {C}onstraint {S}atisfaction: {A} {S}tudy through {D}atalog and
  {G}roup {T}heory.
\newblock {\em {SIAM} Journal on Computing}, 28(1):57--104, 1998.

\bibitem{Fujishige:91}
S.~Fujishige.
\newblock {\em Submodular Functions and Optimization}.
\newblock North-Holland, 1991.

\bibitem{FujishigeMurota:00}
S.~Fujishige and K.~Murota.
\newblock Notes on {L}-/{M}-convex functions and the separation theorems.
\newblock {\em Math. Program.}, 88:129--146, 2000.

\bibitem{Fujishige:06}
Satoru Fujishige and Satoru Iwata.
\newblock Bisubmodular function minimization.
\newblock {\em SIAM J. Discrete Math.}, 19(4):1065--1073, 2006.

\bibitem{Grotschel:88}
M. Gr\"{o}tschel,  L. Lov\'{a}sz and A. Schrijver.
\newblock Geometric Algorithms and Combinatorial Optimization.
\newblock {\em Springer Heidelberg}, 1988.

\bibitem{Iwata:01}
S. Iwata, L. Fleischer and S. Fujishige.
\newblock A combinatorial strongly polynomial algorithm for minimizing submodular functions.
\newblock {\em J. ACM}, 48:761--777, 2001.


\bibitem{Jonsson:TR11}
P. Jonsson, F. Kuivinen and J. Thapper.
\newblock Min {C}{S}{P} on Four Elements: Moving Beyond Submodularity.
\newblock Tech. rep. {\em arXiv:1102.2880}, February 2011.

\bibitem{Kabadi:90}
Santosh~N. Kabadi and R.~Chandrasekaran.
\newblock On totally dual integral systems.
\newblock {\em Discrete Appl. Math.}, 26:87--104, 1990.

\bibitem{KS:09}
V.~Kolmogorov and A.~Shioura.
\newblock New algorithms for convex cost tension problem with application to
  computer vision.
\newblock {\em Discrete Optimization}, 6(4):378--393, 2009.






\bibitem{KZ10:TRa}
V.~Kolmogorov and S.~\v{Z}ivn\'y.
\newblock The complexity of conservative finite-valued {C}{S}{P}s.
\newblock Tech. rep. {\em arXiv:1008.1555v1}, August 2010.

\bibitem{KZ10:TRb}
V.~Kolmogorov and S.~\v{Z}ivn\'y.
\newblock Generalising tractable {V}{C}{S}{P}s defined by symmetric tournament pair multimorphisms.
\newblock Tech. rep. {\em arXiv:1008.3104v1}, August 2010.

\bibitem{K10:TRc}
V.~Kolmogorov.
\newblock A dichotomy theorem for conservative general-valued {C}{S}{P}s.
\newblock Tech. rep. {\em arXiv:1008.4035v1}, August 2010.

\bibitem{K10:TSv2}
V.~Kolmogorov.
\newblock Submodularity on a tree: Unifying $L^\natural$-convex and bisubmodular functions.
\newblock Tech. rep. {\em arXiv:1007.1229v2}, July 2010.

\bibitem{KrokhinLarose:08}
A.~Krokhin and B.~Larose.
\newblock Maximizing supermodular functions on product lattices, with application to maximum constraint satisfaction.
\newblock {\em SIAM Journal on Discrete Mathematics}, 22(1):312--328, 2008.

\bibitem{Kuivinen:TR}
F.~Kuivinen.
\newblock On the Complexity of Submodular Function Minimisation on Diamonds.
\newblock Tech. rep. {\em arXiv:0904.3183v1}, April 2009.

\bibitem{McCormick:10}
S.~Thomas McCormick and Satoru Fujishige.
\newblock Strongly polynomial and fully combinatorial algorithms for
  bisubmodular function minimization.
\newblock {\em Math. Program., Ser. A}, 122:87--120, 2010.

\bibitem{Murota:98}
K.~Murota.
\newblock Discrete convex analysis.
\newblock {\em Math. Program.}, 83:313--371, 1998.

\bibitem{Murota:IEICE00}
K.~Murota.
\newblock Algorithms in discrete convex analysis.
\newblock {\em IEICE Transactions on Systems and Information}, E83-D:344--352,
  2000.

\bibitem{Murota:book}
K.~Murota.
\newblock {\em Discrete Convex Analysis}.
\newblock SIAM Monographs on Discrete Mathematics and Applications, Vol. 10,
  2003.

\bibitem{Murota:SIAM03}
K.~Murota.
\newblock On steepest descent algorithms for discrete convex functions.
\newblock {\em SIAM J. Optimization}, 14(3):699--707, 2003.

\bibitem{Nakamura:88}
M.~Nakamura.
\newblock A characterization of greedy sets: universal polymatroids ({I}).
\newblock In {\em Scientific Papers of the College of Arts and Sciences},
  volume 38(2), pages 155--167. The University of Tokyo, 1998.

\bibitem{Qi:88}
Liqun Qi.
\newblock Directed submodularity, ditroids and directed submodular flows.
\newblock {\em Mathematical Programming}, 42:579--599, 1988.

\bibitem{Schaefer78:complexity}
T.~Schaefer.
\newblock The {C}omplexity of {S}atisfiability {P}roblems.
\newblock In {\em Proceedings
  of the 10th {A}nnual {A}{C}{M} {S}ymposium on {T}heory of {C}omputing
  ({S}{T}{O}{C}'78)}, pages 216--226, 1978.

\bibitem{Schrijver:00}
A. Schrijver.
\newblock A combinatorial algorithm minimizing submodular functions in strongly polynomial time.
\newblock {\em J. Combin. Theory Ser. B}, 80:346--355, 2000.

\bibitem{Takhanov10:stacs}
Rustem Takhanov.
\newblock A {D}ichotomy {T}heorem for the {G}eneral {M}inimum {C}ost
  {H}omomorphism {P}roblem.
\newblock In {\em Proceedings of the 27th International Symposium on
  Theoretical Aspects of Computer Science (STACS'10)}, pages 657--668, 2010.

\bibitem{Topkis:78}
Donald~M. Topkis.
\newblock Minimizing a submodular function on a lattice.
\newblock {\em Operations Research}, 26(2):305--321, 1978.





























\end{thebibliography}

\end{document}